\documentclass[hyperref]{article}

\usepackage{graphicx} 

\usepackage{float} 

\usepackage{amssymb}

\usepackage{indentfirst}

\usepackage{framed}

\usepackage{multicol}

\usepackage{color}

\usepackage{amsmath}

\usepackage{mathtools}

\usepackage[thmmarks,amsmath]{ntheorem} 

\usepackage{bm} 

\usepackage[colorlinks,linkcolor=cyan,citecolor=cyan]{hyperref}

\usepackage{extarrows}

\usepackage{abstract}

\usepackage{authblk}

\usepackage{enumerate}

\usepackage[hang,flushmargin]{footmisc} 

\usepackage[square,comma,sort&compress,numbers]{natbib} 

\usepackage{mathrsfs} 

\usepackage[font=footnotesize,skip=0pt,textfont=rm,labelfont=rm]{caption,subcaption}

\usepackage{booktabs} 

\usepackage{tocloft}

{
	
	\theoremstyle{nonumberplain}
	
	\theoremheaderfont{\bfseries}
	
	\theorembodyfont{\normalfont}
	
	\theoremsymbol{\mbox{$\Box$}}
	
	\newtheorem{proof}{Proof}
	
}

\DeclareMathOperator{\Span}{span}

\DeclareMathOperator{\Supp}{supp}

\DeclareMathOperator{\Wt}{wt}
\usepackage{theorem}
\newtheorem{theorem}{Theorem}

\newtheorem{lemma}{Lemma }

\newtheorem{definition}{Definition}

\newtheorem{prop}{Proposition}

\newtheorem{example}{Example}

\newtheorem{corollary}{Corollary}

\newtheorem{remark}{Remark}

{
	
	\theoremheaderfont{\bfseries}
	
	\theorembodyfont{\normalfont}
	
	\newtheorem{algorithm}{Algorithm }
	
	\graphicspath{{figure/}}                                 
	
	\usepackage[a4paper]{geometry}                           
	
	\geometry{left=2.0cm,right=2.0cm,top=2.0cm,bottom=2.0cm} 
	
	\linespread{1.0}                               
	
	\title{\textbf{Some Results on Locally Repairable Codes with Minimum Distance $7$ and Locality $2$}}
	\date{}
	
	\author[1]{Yuan Gao \thanks{Email: 51205500063@stu.ecnu.edu.cn.}}
	\author[1]{Siman Yang \thanks{Email: smyang@math.ecnu.edu.cn(Corresponding author).}}
	\affil[1]{\footnotesize
		 Department of Mathematics, East China Normal University, Shanghai, 200241, China}
	
\begin{document}	
		\maketitle
		\small
		\begin{multicols}{2}
			\begin{abstract}
				{\bf Locally repairable codes(LRCs) play 
				important roles in distributed storage systems(DSS). LRCs with small 
				locality have their own advantages since fewer available symbols are 
				needed in the recovery of erased symbols. 
				In this paper, we prove an upper bound on the dimension of LRCs with minimum distance $d\geq 7$.
An upper bound  on the length of almost optimal LRCs with 
				$d=7$, $r=2$ at $q^2+q+3$ is proved. Then based on the $t$-spread structure,
				we give an 
				algorithm to construct almost optimal LRCs with $d=7$, $r=2$ and length $n\geq 3\lceil\frac{\sqrt{2}q}{3}\rceil$ when $q\geq 4$, whose dimension attains the aforementioned upper bound. \rm}\\ \par\
			\bf\emph{Index Terms}-Locally repairable codes, distributed storage systems, $t$-spreads.\rm
			\end{abstract}

			\section{Introduction}
		To meet the growing needs for efficient and reliable cloud storage systems and big data storage systems, distributed storage systems have been widely used. In distributed storage systems, a data is partitioned and stored in different
		data storage devices. Inevitably, data loss may occur. Erasure codes are introduced to prevent data loss, such as locally repairable codes.
			The notions of \emph{the locality of code symbols} and the formal definition of \emph{locally repairable codes} were introduced in \cite{gopalan2012locality} and \cite{papailiopoulos2014locally}. The bounds and constructions of locally repairable codes have attracted the attention of many researchers in 
			recent years.
			
			Let $q$ be a prime power. 
			An $[n,k,d]_q$ linear code $\mathcal{C}$ is called a locally repairable code with locality $r$ if for every $i\in \{1,2,\dots,n\}$, there exists a codeword $\boldsymbol{c'}$ in its dual code $\mathcal{C}^{\bot}$ such that $\Wt(\boldsymbol{c'})\leq r+1$ and $i\in \Supp(\boldsymbol{c'})$.
			The well-known Singleton-like bound
			\begin{eqnarray}\label{eq1}
				d\leq n-k-\lceil\frac{k}{r}\rceil+2
			\end{eqnarray}
			was proved in \cite{gopalan2012locality}.
			
			LRCs whose parameters attain bound \eqref{eq1} are called  optimal LRCs.
			LRCs with parameters satisfying $d=n-k-\lceil\frac{k}{r}\rceil+1$ are called almost optimal LRCs. 
			In the works \cite{guruswami2019long},\cite{chen2020improved}, the upper bounds on the length of optimal LRCs with fixed $d,r$ were proved.  
			There were many constructions of optimal LRCs,
			e.g.,  \cite{guruswami2019long}, \cite{chen2020improved}, \cite{tamo2014family}, \cite{li2018optimal}, \cite{luo2018optimal}, \cite{jin2019explicit}, \cite{xing2018construction}.
			In \cite{guruswami2019long}, LRCs with $d=3,4$ and unbounded length were constructed, and for any $5\leq d\leq r+2$, optimal LRCs with length at least $cq^{1+\frac{1}{\lfloor(d-3)/2\rfloor}}$ were constructed, where $c$ is a constant that only depends on $d$ and $r$.
			By polynomial evaluation, \cite{tamo2014family} constructed optimal LRCs with $n\leq q$. It is worth noting that the  constructions in \cite{tamo2014family} only require $r,k$ to satisfy $r\leq k$ and $\frac{k}{r}\leq \frac{n}{r+1}$. By the algebraic structure of elliptic curves,   optimal LRCs with length up to $q+\sqrt{q}$ were constructed in \cite{li2018optimal}. Optimal cyclic LRCs with $d=3,4$ and unbounded length were constructed in \cite{luo2018optimal}. In \cite{jin2019explicit}, via binary constant weight codes, optimal LRCs with $d=5,6$ and length greater than $cq^{2}$ were explicitly constructed, where $c$ is a constant that only depends on $d,r$.
			In \cite{xing2018construction}, the connections between optimal LRCs with Vandermonde-like parity-check matrices and hypergraphs were established. Practical algorithms were given to obtain the desired parity-check matrices of optimal LRCs, whose length is at least $cq^{1+\frac{1}{\lfloor(d-3)/2\rfloor}}$, where $c$ is a constant that only depends on $d$ and $r$.
		In \cite{chen2020improved}, the optimal LRCs with $d=6$, $r=2$ and disjoint local repair groups were completely characterized, then optimal LRCs with lengths $3(q+1),3(q+\sqrt{q}+1)$ and $3(2q-4)$ were explicitly constructed.
%
			However, as will be shown in Lemma \ref{lem2},  there is no optimal LRC with $d=7$, $r=2$ and disjoint local repair groups.
			Thus, we attempt to construct almost optimal LRCs with $d=7$, $r=2$ and disjoint local repair groups. Some upper bounds are also derived in the process.
			
			The remainder of the paper is organized as follows. In Section 2, we review some definitions and necessary results. In Section 3, we give our main results as follows:
			
			(1)An explicit upper bound on the dimension of LRCs with minimum distance $d\geq 7$ and disjoint local repair groups. This upper bound outperforms the explicit upper bound proved in \cite{wang2019bounds} when $r=2,3$.
			
			(2)An upper bound on the length of almost optimal LRCs with minimum distance $d=7$, locality $r=2$ and disjoint local repair groups at $q^2+q+3$. This upper bound is better than that proved in \cite{guruswami2019long} and \cite{chen2020improved} when $d=7$, $r=2$.
			
			(3)An algorithm to construct almost optimal LRCs with minimum distance $7$, locality $2$ and disjoint local repair groups, whose length is at least $3\lceil\frac{\sqrt{2}}{3}q\rceil$ when $q\geq 4$. It is worth noting that these LRCs attain the upper bound on dimension in Lemma \ref{lem3}.
			\section{Preliminaries}
			
			\subsection{Notations and definitions}
			We recall some notations and definitions.
			\begin{itemize}
				\item Let $\lfloor a \rfloor,\lceil a \rceil$ be the floor function and the ceiling function of $a$, respectively.
				\item Let $\mathbb{F}_q$ be the finite field with $q$ elements, where $q$ is a prime power. 
				\item Let $\mathbb{F}_q^n$ be the $n$-dimensional vector space over $\mathbb{F}_q$.
				\item For any positive integer $n$,
				we denote by $[n]$ the set $\{1,2,\dots,n\}$. 
				\item For any $\boldsymbol{c}=(c_1,c_2,\dots,c_n)\in \mathbb{F}_q^n$,
				the support set of $\boldsymbol{c}$ is defined as
				 $\Supp(\boldsymbol{c})=\{i\in [n]|c_i\neq 0\}$,
				 the hamming weight of $\boldsymbol{c}$ is defined as $\Wt(\boldsymbol{c})=|\Supp(\boldsymbol{c})|$.
			
			\item For any non-trivial vector space $V$ over $\mathbb{F}_q$, we denote the set of all one-dimensional subspaces of $V$ by $\left[V \atop 1 \right]$.
				
			\end{itemize}

			Now, we give the formal definition of locally repairable codes.
			\begin{definition}\label{defn1}
				A $q$-ary $[n,k,d]$ linear code $\mathcal{C}$ is said to be a locally repairable code with 
				locality $r$ if for any $i\in [n]$, there exists a set $R_i\subseteq [n]\backslash \{i\}$
				such that $|R_i|\leq r$ and for any codeword $(c_1,c_2,\dots,c_n)\in \mathcal{C}$, 
				$$                                                            
				c_i=\sum_{j\in R_i}{k_jc_j},
				$$
				where $k_j\in \mathbb{F}_q\backslash \{0\}$ are fixed elements. $R_i\cup \{i\}$ is called a local repair group. 
				We denote its parameters by $(n,k,d,r)_q$.
				\end{definition}
 If some local repair groups form a partition of $[n]$, 
and each of them has cardinality $(r+1)$, then we say that $\mathcal{C}$ has disjoint local
repair 
groups.

An LRC has disjoint local repair groups if and only if it has an equivalent parity-check matrix of the following form:
\setlength{\arraycolsep}{0.8pt}
\begin{equation*}
	\left[\begin{array}{cccc|cccc|c|cccc}
		1&1&\dots&1&0&0&\dots&0&\dots&0&0&\dots&0\\
		0&0&\dots&0&1&1&\dots&1&\dots&0&0&\dots&0\\
		\vdots&\vdots&\vdots&\vdots&\vdots&\vdots&\vdots&\vdots&\vdots&\vdots&\vdots&\vdots&\vdots\\
		0&0&\dots&0&0&0&\dots&0&\dots&1&1&\dots&1\\
		\hline
		\boldsymbol{u}_1^{(1)}&\dots&\boldsymbol{u}_r^{(1)}&\boldsymbol{0}&
		\boldsymbol{u}_1^{(2)}&\dots&\boldsymbol{u}_r^{(2)}&\boldsymbol{0}&\dots&\boldsymbol{u}_1^{(L)}&\dots&\boldsymbol{u}_r^{(L)}&\boldsymbol{0}\\
	\end{array}\right]
\end{equation*}
			\subsection{The LRCs with minimum distance $7$, locality $2$ and disjoint local repair groups}
			In this paper, we focus on LRCs with minimum distance $7$, locality $2$ and disjoint local repair groups. There exists no LRC attaining the Singleton-like bound \eqref{eq1} in this case.  
			\begin{lemma}[\cite{guruswami2019long}, Lemma 2.2]\label{lem1}
				Let $n,k,d,r$ be positive integers.
				
				If $(r+1)|n$ and 
				$n-k-\lceil\frac{k}{r}\rceil+2=d$, 
				then 
				\begin{equation}\label{eq2}
					n-k-\frac{n}{r+1}=d-2-\lfloor\frac{d-2}{r+1}\rfloor
				\end{equation}
				Also, equality \eqref{eq2} along with 
				$(r+1)|n$ 
				and $d-2\not\equiv r\mod(r+1)$ imply 
				$n-k-\lceil\frac{k}{r}\rceil+2=d.$
			\end{lemma}
			
			\begin{lemma}[\cite{xing2018construction}, Remark 1]\label{lem2}
				Let $n,k,d,r$ be positive integers and $q$ be a prime power. If $d-2\equiv r\mod (r+1)$, 
				then there is 
				no $(n,k,d,r)_q$  LRC 
				satisfying $(r+1)|n$ and $n-k-\lceil\frac{k}{r}\rceil+2=d$.
			\end{lemma}
		\begin{remark}
			Due to this
			phenomenon, an $(n,k,d,r)_q$ LRC satisfying $d-2\equiv r\mod (r+1)$, $(r+1)|n$, and $d=n-k-\lceil\frac{k}{r}\rceil+1$ is also called an optimal LRC in \cite{guruswami2019long} and \cite{xing2018construction}. However, in this paper, we call it an almost optimal LRC since its minimum distance differs one from the
			Singleton-like bound \eqref{eq1}.
		\end{remark}
		We review a recent result in \cite{cai2022bound}. It will be useful to determine the minimum distance of the LRCs in our construction.
			\begin{theorem}[\cite{cai2022bound}, Theorem 1]\label{thm1}
				Let C be an optimal $[n, k, d]_q$ linear code with all-symbol $(r,\delta)$-locality. Assume $n = m(r +\delta-1), k = ur, u\geq 2$.
				
				If $2|r$ and $m\geq u+1$ then
				$$q\geq \phi
				\Bigg(\bigg((\frac{k}{r}+1)\lfloor\frac{2r+2\delta-2}{r}\rfloor-1)
				\bigg)^{\frac{2}{r}}\Bigg),$$
				where $\phi(x)$ is the smallest prime power greater or equal to $x$.
				
				If $2\nmid r$ and $m\geq u + 2$, then
				$$q\geq \phi\bigg((\frac{k}{r}^{\frac{2}{r+1}})\bigg)$$
			\end{theorem}
			When $\delta=2$ in Theorem \ref{thm1}, $\mathcal{C}$ is an LRC with locality $r$ as defined in Definition \ref{defn1}. We then have the following corollary.
			\begin{corollary}\label{cor1}
				Let $\mathcal{C}$ be an $(n,k,d,r)_q$ LRC.
				If $r=2$, $n=3L\textgreater q+4$ and $k=2(L-2)$, then the minimum distance $d\leq 7$.
			\end{corollary}
	\begin{proof}
	By the inequality \eqref{eq1},$$d\leq 3L-2(L-2)-\lceil\frac{2(L-2)}{2}\rceil+2=8.$$ If $d=8$, then $\mathcal{C}$ is an optimal LRC satisfying the condition of Theorem \ref{thm1}. By Theorem \ref{thm1}, we have $$q\geq \phi(3(L-2+1)-1)\geq 3L-4=n-4,$$ which is a contradiction. Hence $d\leq 7$.
\end{proof}
			Thus, if we have an almost optimal LRC with parameters $(n=3L,k=2L-4,d=7,r=2)_q$ and $n\textgreater q+4$, then it obtains the best possible minimum distance when its other parameters and the field size $q$ are unchanged.

			\section{Main results}
			\subsection{An upper bound on the dimension of LRCs with disjoint local repair groups and $d\geq 7$ }
			By analysing the parity-check matrix, we have the following upper bound on the dimension of LRCs with minimum distance $d\geq 7$ and disjoint local repair groups.
			\begin{lemma*}\label{lem3}	
				Let $\mathcal{C}$ be an LRC with disjoint local repair groups and parameters $$(n=L(r+1),k=Lr-u,d\geq 7,r)_q,$$ then we have
				
				\begin{equation}\label{eq3}
					k\leq \frac{rn}{r+1}-\lceil\log_q\big(q+(q-1)q(\frac{r}{2}n-r)\big)\rceil
				\end{equation}
			\end{lemma*}
			\begin{proof} 
				$\mathcal{C}$ has an equivalent parity-check matrix $P$ of the following form. 
				\begin{equation*}
				P=	\left[\begin{array}{cccc|cccc|c|cccc}
						1&1&\dots&1&0&0&\dots&0&\dots&0&0&\dots&0\\
						0&0&\dots&0&1&1&\dots&1&\dots&0&0&\dots&0\\
						\vdots&\vdots&\vdots&\vdots&\vdots&\vdots&\vdots&\vdots&\vdots&\vdots&\vdots&\vdots&\vdots\\
						0&0&\dots&0&0&0&\dots&0&\dots&1&1&\dots&1\\
						\hline
						\boldsymbol{u}_1^{(1)}&\dots&\boldsymbol{u}_r^{(1)}&\boldsymbol{0}&
						\boldsymbol{u}_1^{(2)}&\dots&\boldsymbol{u}_r^{(2)}&\boldsymbol{0}&\dots&\boldsymbol{u}_1^{(L)}&\dots&\boldsymbol{u}_r^{(L)}&\boldsymbol{0}\\
					\end{array}\right]
				\end{equation*}
				Let $\boldsymbol{h}_{i}^{(t)}$ be the $\big((t-1)(r+1)+i\big)$-th column of the above matrix, $\boldsymbol{u}_{r+1}^{(t)}=\boldsymbol{0}\in \mathbb{F}_q^{u}$ for $t\in [L]$, and  $\boldsymbol{u}_{i,j}^{(t)}\triangleq \boldsymbol{u}_{j}^{(t)}-\boldsymbol{u}_{i}^{(t)}$ for $i\neq j\in [r+1], t\in [L].$ We have the following claim.
				
				\textbf{Claim 1}:Let $S_1=\{(i,j,t)|i\textless j\in [r+1]\backslash \{1\},t=1\}$, $S_2=\{(i,j,t)|i\textless j\in [r+1],t\in [L]\backslash \{1\}\}$, then	$\{\Span\{\boldsymbol{u}_{1,2}^{(1)},\boldsymbol{u}_{i,j}^{(t)}\}|(i,j,t)\in S_1\cup S_2\}$
				is a set of distinct two-dimensional subspaces of $\mathbb{F}_q^u$ and these subspaces pairwisely intersect in the
				 one-dimensional subspace $\Span\{\boldsymbol{u}_{1,2}^{(1)}\}$ .  
				
				To prove Claim 1, we need to show that for any $(i_1,j_1,t_1) \neq (i_2,j_2,t_2)\in S_1\cup S_2$,
				$\boldsymbol{u}_{i_1,j_1}^{(t_1)},\boldsymbol{u}_{i_2,j_2}^{(t_2)},\boldsymbol{u}_{1,2}^{(1)}$ are linear independent.
				
				Assume that $k_1\boldsymbol{u}_{i_1,j_1}^{(t_1)}+k_2\boldsymbol{u}_{i_2,j_2}^{(t_2)}+k_3\boldsymbol{u}_{1,2}^{(1)}=\boldsymbol{0}$, where $(k_1,k_2,k_3)\in \mathbb{F}_q^3$, then we have
				\begin{footnotesize}
					$$k_1(\boldsymbol{h}_{j_1}^{(t_1)}-\boldsymbol{h}_{i_1}^{(t_1)})+k_2(\boldsymbol{h}_{j_2}^{(t_2)}-\boldsymbol{h}_{i_2}^{(t_2)})+k_3(\boldsymbol{h}_{2}^{(1)}-\boldsymbol{h}_{1}^{(1)})=\boldsymbol{0}.$$\end{footnotesize}
				If $k_3\neq 0$, then the coefficient of $\boldsymbol{h}_{1}^{(1)}$ in above equation is non-zero. Note that any six columns of $P$ are linear independent. This leads to a contradiction.
				
				If $k_3=0$, $k_1\neq 0$, then the coefficient of $\boldsymbol{h}_{j_1}^{(t_1)}$ or $\boldsymbol{h}_{i_1}^{(t_1)}$ in above equation is non-zero, which is a contradiction.
				
				If $k_3=0$, $k_2\neq 0$, then the coefficient of $\boldsymbol{h}_{j_2}^{(t_2)}$ or $\boldsymbol{h}_{i_2}^{(t_2)}$ in above equation is non-zero, which is a contradiction. 
				
				Thus, $k_1=k_2=k_3=0$, $\boldsymbol{u}_{i_1,j_1}^{(t_1)},\boldsymbol{u}_{i_2,j_2}^{(t_2)},\boldsymbol{u}_{1,2}^{(1)}$ are linear independent.
				
				By Claim 1, there exist $\big(L\frac{r(r+1)}{2}-r\big)$ distinct two-dimensional subspaces containing $\Span\{\boldsymbol{u}_{1,2}^{(1)}\}$ in $\mathbb{F}_q^u$. Note that there are totally $\frac{q^u-q}{q^2-q}$ two-dimensional subspaces containing $\Span\{\boldsymbol{u}_{1,2}^{(1)}\}$ in $\mathbb{F}_q^u$. We have		
				\begin{equation}\label{eq4}
					\frac{q^u-q}{q^2-q}\geq L\frac{r(r+1)}{2}-r,
				\end{equation}
				$$q^u\geq q+(q-1)q(\frac{r}{2}n-r),$$
				$$k\leq \frac{rn}{r+1}-\log_q\big(q+(q-1)q(\frac{r}{2}n-r)\big).$$
				
				Since $k$ and $\frac{rn}{r+1}$ are positive integers, we have
				\begin{equation*}
					k\leq \frac{rn}{r+1}-\lceil\log_q\big(q+(q-1)q(\frac{r}{2}n-r)\big)\rceil.
				\end{equation*}
			\end{proof}
			\begin{remark}\label{rem2}
				By Theorem 11 in \cite{wang2019bounds}, under the setting of Lemma \ref{lem3}, we have
				\begin{footnotesize}
				\begin{equation*}
					k\leq \frac{rn}{r+1}- \log_q\big(1+(q-1)\frac{rn}{2}+\frac{(r-1)r(q-1)(q-2)n}{6}\big).
				\end{equation*}
			\end{footnotesize}
				Our upper bound \eqref{eq3} is tighter than this upper bound when $r=2$ or $3$, $q\geq 2$. The LRC in Example \ref{exma1} attains the bound \eqref{eq3}, but does not attain the above bound.
			\end{remark}
		\subsection{An upper bound on the length of almost optimal LRCs with minimum distance $7$, locality $2$ and disjoint local repair groups}
		\begin{lemma}\label{lem4}	
			Let $\mathcal{C}$ be an almost optimal LRC with disjoint local repair groups and parameters  $$(n=3L,k=2L-4,d=7,r=2)_q,$$
		then
			\begin{equation}\label{eq5}
				n\leq 3+q(q+1).
			\end{equation}
		\end{lemma}
		\begin{proof}
			By the inequality \eqref{eq4} in the proof of Lemma \ref{lem3}, $$\frac{q^u-q}{q^2-q}\geq L\frac{r(r+1)}{2}-r,$$ where $u\triangleq Lr-k=4$.
			
		   Thus	$n\leq q^2+q+3.$
		\end{proof}

		\begin{remark}\label{rem3}
			(1) Let $\mathcal{C}$ be an LRC with disjoint local repair groups. If its parameters $(n,k,d,r)_q$ satisfy the equation $\eqref{eq2}$ and $d\geq 5$, 
			
			then by  Theorem 3.2 in \cite{guruswami2019long}, we have
			\begin{equation*}
				n\leq \begin{cases}
					\frac{r+1}{r}\times \frac{d-a}{4(q-1)}\times q^{\frac{4(d-2)}{d-a}}&d\equiv 1,2 \mod 4\\
					\frac{r+1}{r}\times (\frac{d-a}{4(q-1)}\times q^{\frac{4(d-3)}{d-a}}+1)&d\equiv 3,4 \mod 4\\
				\end{cases},
			\end{equation*}
		where $a\triangleq d-4(\lceil\frac{d}{4}\rceil-1)$.
		
			by Lemma 3 in \cite{chen2020improved}, we have
			$$n\leq \frac{2}{r}\times \frac{q^{d-2-\lfloor\frac{d-2}{r+1}\rfloor}-1}{q-1}.$$

			When $d=7,r=2$, the above two upper bounds become $n\leq  \frac{3q^4}{2(q-1)}$ and $n\leq \frac{(q^4-1)}{(q-1)}$, respectively.
			The upper bound \eqref{eq5} is better than the above two upper bounds in this case.
%
		\end{remark}
			\subsection{Almost optimal LRCs with minimum distance $7$ and locality $2$ based on $t$-spread}
			In this subsection, we will construct almost optimal LRCs with minimum distance $7$, locality $2$ and disjoint local repair groups.
			The following lemma is useful.
			\setlength{\arraycolsep}{1.2pt}
			
			\begin{lemma}\label{lem5}
				Let $q$ be a prime power. $L\geq 3$ is an integer. There is an LRC with disjoint local repair groups and parameters 
				$[n=3L,k=2L-4,7\leq d\leq 
				8,2]_q$
				if and only if there is a vector sequence
				$\boldsymbol{u}_1^{(i)},\boldsymbol{u}_2^{(i)}\in
				\mathbb{F}_q^4, i\in[L]$
				satisfying the following conditions
				\begin{enumerate}[\textbf{C}.1]
					\item 
					$\dim(\Span\{\boldsymbol{u}_1^{(i)},\boldsymbol{u}_2^{(i)}\})=2$ for any $i\in [L]$.
					\item $\Span\{\boldsymbol{u}_1^{(i)},\boldsymbol{u}_2^{(i)}\}\cap 
					\Span\{\boldsymbol{u}_1^{(j)},\boldsymbol{u}_2^{(j)}\}=\{\boldsymbol{0}\}$ for any $i\textless j\in [L]$.
					\item 
					$\dim(\Span\{\boldsymbol{u}_a^{(i)},\boldsymbol{u}_b^{(j)},\boldsymbol{u}_c^{(t)}\})=3$  for any $i\textless j\textless t\in[L]$, 
					$a,b,c\in \{0,1,2\}$.
				\end{enumerate}
				where $\boldsymbol{u}_0^{(i)}\triangleq \boldsymbol{u}_1^{(i)}-\boldsymbol{u}_2^{(i)}.$
			\end{lemma} 
			\begin{proof}
				Let $H$ be an $(L+4)\times 3L$ matrix over $\mathbb{F}_q$ of the following form.
				\setcounter{MaxMatrixCols}{30}
				\setlength{\arraycolsep}{0.6pt}
				\begin{equation*}
					H=\left[\begin{array}{ccc|ccc|c|ccc}
						1&1&1&0&0&0&\cdots&0&0&0\\
						0&0&0&1&1&1&\cdots&0&0&0\\
						\vdots&\vdots&\vdots&\vdots&\vdots&\vdots&\vdots&\vdots&\vdots&\vdots\\
						0&0&0&0&0&0&\cdots&1&1&1\\
						\hline
						\boldsymbol{u}_1^{(1)}&\boldsymbol{u}_2^{(1)}&\boldsymbol{0}&
						\boldsymbol{u}_1^{(2)}&\boldsymbol{u}_2^{(2)}&\boldsymbol{0}&\dots&\boldsymbol{u}_1^{(L)}&\boldsymbol{u}_2^{(L)}&\boldsymbol{0}\\
					\end{array}\right]
				\end{equation*}
				Let $\boldsymbol{h}_a^{(i)}(i\in [L],a\in [3])$ be the $\big(3(i-1)+a\big)$-th column of $H$.
				We divide the proof of Lemma \ref{lem5} into
				two parts.
				
				\emph{Necessity}:
				If there is an LRC with disjoint local repair groups and parameters
				$(3L,2L-4,7\leq d\leq 8,r=2)_q$, then it has
				an equivalent parity-check matrix $H$ and any six columns of $H$ are linear independent.
				
				Then we have
				\begin{enumerate}[(1)]
					\item $\boldsymbol{h}_1^{(i)},\boldsymbol{h}_2^{(i)},\boldsymbol{h}_3^{(i)},\boldsymbol{h}_1^{(j)},\boldsymbol{h}_2^{(j)},\boldsymbol{h}_3^{(j)}$ are linear independent for any $i\textless j\in [L]$
					
					\item $\boldsymbol{h}_a^{(i)},\boldsymbol{h}_b^{(i)},\boldsymbol{h}_c^{(j)},\boldsymbol{h}_d^{(j)},\boldsymbol{h}_e^{(t)},\boldsymbol{h}_f^{(t)}$ are linear independent for any $i\textless j\textless t\in [L]$ and
					$a\neq b,c\neq 
					d,e\neq f\in [3]$.
				\end{enumerate}	
				
				We identify the above vector groups with matrices, then by elementary transformations, we have
				\begin{enumerate}[(1)]
					
					\item 
					$\boldsymbol{h}_1^{(i)}-\boldsymbol{h}_3^{(i)},\boldsymbol{h}_2^{(i)}-\boldsymbol{h}_3^{(i)},\boldsymbol{h}_3^{(i)},\boldsymbol{h}_1^{(j)}-\boldsymbol{h}_3^{(j)},\boldsymbol{h}_2^{(j)}-\boldsymbol{h}_3^{(j)},\boldsymbol{h}_3^{(j)}$
					are linear independent for any $i\textless j\in [L]$.
					
					\item 
					$\boldsymbol{h}_a^{(i)}-\boldsymbol{h}_b^{(i)},\boldsymbol{h}_b^{(i)},\boldsymbol{h}_c^{(j)}-\boldsymbol{h}_d^{(j)},\boldsymbol{h}_d^{(j)},\boldsymbol{h}_e^{(t)}-\boldsymbol{h}_f^{(t)},\boldsymbol{h}_f^{(t)}$
					are linear independent for any $i\textless j\textless t\in [L]$ and
					$a\neq b,c\neq 
					d,e\neq f\in [3]$.
				\end{enumerate}	
				Note that the first $L$ components of $\boldsymbol{h}_a^{(i)}-\boldsymbol{h}_b^{(i)}(a\neq b\in [3],i\in [L])$ are all zeros. We have
				\begin{enumerate}[(1)]
					\item 
					$\boldsymbol{u}_1^{(i)},\boldsymbol{u}_2^{(i)},\boldsymbol{u}_1^{(j)},\boldsymbol{u}_2^{(j)}$
					are linear independent for any $i\textless j\in [L]$.
					
					\item 
					$\boldsymbol{u}_a^{(i)},\boldsymbol{u}_b^{(j)},\boldsymbol{u}_c^{(t)}$
					are linear independent for any $i\textless j\textless t\in [L]$ and
					$a,b,c\in \{0,1,2\}$.
				\end{enumerate}	
				Thus the vector sequence $\boldsymbol{u}_1^{(i)},\boldsymbol{u}_2^{(i)},i=1,2,\dots,L$ must satisfy the conditions \textbf{C}.1,\textbf{C}.2 and \textbf{C}.3. The necessity is proved.
				
				\emph{sufficiency}:
				Assume that the vector sequence $\boldsymbol{u}_1^{(i)},\boldsymbol{u}_2^{(i)}\in
				\mathbb{F}_q^4, i\in[L]$ satisfies the conditions  \textbf{C}.1,\textbf{C}.2,\textbf{C}.3. Then we have
				\begin{enumerate}[(1)]
					\item 
					$\boldsymbol{u}_1^{(i)},\boldsymbol{u}_2^{(i)},\boldsymbol{u}_1^{(j)},\boldsymbol{u}_2^{(j)}$
					are linear independent for any $i\textless j \in [L]$. 
					
					\item 
					$\boldsymbol{u}_a^{(i)},\boldsymbol{u}_b^{(j)},\boldsymbol{u}_c^{(t)}$
					are linear independent for any $i\textless j\textless t\in [L]$ and
					$a,b,c\in \{0,1,2\}$.
				\end{enumerate}	
			We put the above vectors into their  corresponding positions of $H$. Note that the first $L$ components of $\boldsymbol{h}_a^{(i)}-\boldsymbol{h}_b^{(i)}(a\neq b\in [3],i\in [L])$ are all zeros. 
			   Then	we have
				\begin{enumerate}[(1)]
					
					\item 
					$\boldsymbol{h}_1^{(i)}-\boldsymbol{h}_3^{(i)},\boldsymbol{h}_2^{(i)}-\boldsymbol{h}_3^{(i)},\boldsymbol{h}_3^{(i)},\boldsymbol{h}_1^{(j)}-\boldsymbol{h}_3^{(j)},\boldsymbol{h}_2^{(j)}-\boldsymbol{h}_3^{(j)},\boldsymbol{h}_3^{(j)}$
					are linear independent for any $i\textless j\in [L]$.
					
					\item 
					$\boldsymbol{h}_a^{(i)}-\boldsymbol{h}_b^{(i)},\boldsymbol{h}_b^{(i)},\boldsymbol{h}_c^{(j)}-\boldsymbol{h}_d^{(j)},\boldsymbol{h}_d^{(j)},\boldsymbol{h}_e^{(t)}-\boldsymbol{h}_f^{(t)},\boldsymbol{h}_f^{(t)}$
					are linear independent for any $i\textless j\textless t\in [L]$ and
					$a\neq b,c\neq 
					d,e\neq f\in [3]$.

				\end{enumerate}	
			    We identify the above vector groups with matrices, by elementary transformations, we have
				\begin{enumerate}[(1)]
					
					\item $\boldsymbol{h}_1^{(i)},\boldsymbol{h}_2^{(i)},\boldsymbol{h}_3^{(i)},\boldsymbol{h}_1^{(j)},\boldsymbol{h}_2^{(j)},\boldsymbol{h}_3^{(j)}$ are linear independent for any $i\textless j\in [L]$. 
					
					\item $\boldsymbol{h}_a^{(i)},\boldsymbol{h}_b^{(i)},\boldsymbol{h}_c^{(j)},\boldsymbol{h}_d^{(j)},\boldsymbol{h}_e^{(t)},\boldsymbol{h}_f^{(t)}$ are linear independent for any $i\textless j\textless t\in [L]$ and
					$a\neq b,c\neq 
					d,e\neq f\in [3]$.
				\end{enumerate}	
				
				Thus, any six columns of $H$ are linear independent, the minimum distance of the LRC with $H$ as a parity-check matrix is at least $7$, the dimension of this LRC is at least $2L-4$. If the dimension is greater than $2L-4$, then there will be a contradiction by the Singleton-like bound \eqref{eq1}. So, its dimension must be $2L-4$. Again by the Singleton-like bound, its minimum distance is at most $8$. The sufficiency is proved.
			\end{proof}
			With Lemma \ref{lem5}, we could construct LRCs with minimum
			distance $7$, locality $2$ and disjoint local repair groups, whose length is at least $3\lceil\frac{\sqrt{2}}{3}q\rceil$. The construction is based on $t$-spread of vector spaces. We briefly review its definition and some related results.
			\begin{definition}[t-spread] 
				Given positive integers $t\leq m$ and $m$-dimensional vector space 
				$V_m(q)$ over 
				$\mathbb{F}_q$. The set of all $t$-dimensional subspaces of $V_m(q)$ is 
				denoted by $\mathcal{G}_q(m,t)$.
				
				If $ \mathcal{S}=\{W_1,W_2,\dots,W_L\}\subseteq \mathcal{G}_q(m,t)$ satisfies $W_a\cap W_b=\{\boldsymbol{0}\}$
				for any $a\neq 
				b\in 
				[L]$, then we call $\mathcal{S}$ a partial $t$-spread of $V_m(q)$ with size $L$.
				
				If also $\cup_{i=1}^{L}{W_i}
				=V_m(q)$, then we call $S$ a $t$-spread of $V_m(q)$ with size $L$.
			\end{definition}
			\begin{lemma}[\cite{bu1980partitions}, Lemma  2]\label{lem7}
				If positive integers $t|m$, then there is a $t$-spread of 
				$\mathbb{F}_q^{m}$ with size $\frac{q^m-1}{q^t-1}$.
			\end{lemma}
%
		The following Corollary is an immediate consequence.
			\begin{corollary}\label{cor2}
				There exists a $2$-spread of 
				$\mathbb{F}_q^{4}$ with size $q^2+1$. 
			\end{corollary}
			 With this $2$-spread, we have an algorithm to find 
			$\boldsymbol{u}_1^{(i)},\boldsymbol{u}_2^{(i)},i=1,2\dots,L$ 
			satisfying the condition of Lemma \ref{lem5}.
		
				\begin{algorithm}\label{alg1} 
					\hspace*{\fill}
					
					\textbf{Input} Finite Field $\mathbb{F}_q,q\geq 4$ 
					\begin{enumerate}[\textbf{Step} 1]
						\item (Initialization)
						
						Let $i=0$.
						
					By Corollary \ref{cor2}, $\mathbb{F}_q^4$ has a $2$-spread $\mathcal{P}=\{P_1,P_2,\dots,P_{q^2+1}\}$. 
					
					Let set family $\mathcal{M}=\{A_1,A_2,\dots,A_{q^2+1}\}$, where $A_t=\left[P_t\atop 1\right]$ for any $t\in [q^2+1]$.
					
					Let $\mathcal{M}_0=\mathcal{M}$.
						\item (Choose)
						
						Let $i=i+1$.
						
						Arbitrarily pick a set $A$ from set family $\mathcal{M}$,  
						
						arbitrarily pick three elements from $A$, 
						
						find their bases 
						$\boldsymbol{u}_0^{(i)},\boldsymbol{u}_1^{(i)},\boldsymbol{u}_2^{(i)}$ 
						such 
						that $$\boldsymbol{u}_0^{(i)}=\boldsymbol{u}_1^{(i)}-\boldsymbol{u}_2^{(i)}.$$
						Let $\mathcal{M}=\mathcal{M}\backslash \{A\}$.
						\item (Trim)
						
						If $i\geq 2$ and $\mathcal{M}\neq \varnothing$,
					  
					   \quad	For $B\in \mathcal{M}$, 
						
						\qquad  For $j\in [i-1]$,
						
						\qquad\quad	Let $B=B\backslash 
						\bigcup\limits_{a,b\in \{0,1,2\}} \left[\Span\{\boldsymbol{u}_a^{(i)},\boldsymbol{u}_b^{(j)}\}\atop 1\right].$
						
%
%

						\qquad end for
						
							\quad\quad	If 
						$|B|\textless 3$,
						
						\quad\qquad	Let $\mathcal{M}=\mathcal{M}\backslash \{B\}$.
						
						\quad\quad end if
					 
					 \quad  end for
					 
						end if

					   If $\mathcal{M}\neq \varnothing$, 
						
				     	\quad	go back to Step 2.
						
						 end if
						
					  If $\mathcal{M}=\varnothing$, 
						
						\quad let $L=i$,
						
						\quad	 \textbf{Output} 
						 $\{\boldsymbol{u}_1^{(1)},\boldsymbol{u}_2^{(1)},\dots,\boldsymbol{u}_1^{(L)},\boldsymbol{u}_2^{(L)}\}$.
						 
						end if

					\end{enumerate}
				\end{algorithm}
			Since $|\mathcal{M}_0|\leq q^2+1$, Algorithm \ref{alg1} stops when $i=i_0$ for some $i_0 \in [q^2+1]\backslash\{1\}$.
			\begin{prop}\label{prop1}
				The output of Algorithm \ref{alg1} is a set of vectors
				$\{\boldsymbol{u}_1^{(1)},\boldsymbol{u}_2^{(1)},\boldsymbol{u}_1^{(2)},\boldsymbol{u}_2^{(2)},
				\dots,\boldsymbol{u}_1^{(L)},\boldsymbol{u}_2^{(L)}\}$. If $L\geq 3$, then it satisfies the condition of Lemma \ref{lem5}.
			\end{prop}
			\begin{proof}
			  
				Note that 
				$\Span\{\boldsymbol{u}_1^{(1)},\boldsymbol{u}_2^{(1)}\},\dots,\Span\{\boldsymbol{u}_1^{(L)},\boldsymbol{u}_2^{(L)}\}$ are distinct two-dimensional subspaces from the 
				$2$-spread 
				$\mathcal{P}$ by Step 2 of Algorithm \ref{alg1}.  $\{\boldsymbol{u}_1^{(1)},\boldsymbol{u}_2^{(1)},\boldsymbol{u}_1^{(2)},\boldsymbol{u}_2^{(2)},
				\dots,\boldsymbol{u}_1^{(L)},\boldsymbol{u}_2^{(L)}\}$ satisfies the conditions \textbf{C}.1 and \textbf{C}.2 in Lemma \ref{lem5}.
				
				We then verify that it satisfies \textbf{C}.3. For any positive integers $1\leq \alpha\textless\beta\textless \gamma\leq L$ and $a,b,c\in
				\{0,1,2\}$, we need to prove 
				that non-zero vectors 
				$\boldsymbol{u}_a^{(\alpha)},\boldsymbol{u}_b^{(\beta)},\boldsymbol{u}_c^{(\gamma)}$ are linear independent.
				
					After every Step 3 in Algorithm \ref{alg1}, we call the 
				past Step 2 and Step 3 the $i$-th round. Note that $\boldsymbol{u}_a^{(\alpha)}\in \Span\{\boldsymbol{u}_1^{(\alpha)},\boldsymbol{u}_2^{(\alpha)}\},\boldsymbol{u}_b^{(\beta)}\in \Span\{\boldsymbol{u}_1^{(\beta)},\boldsymbol{u}_2^{(\beta)}\}$
				and 
				$\Span\{\boldsymbol{u}_1^{(\alpha)},\boldsymbol{u}_2^{(\alpha)}\}, \Span\{\boldsymbol{u}_1^{(\beta)},\boldsymbol{u}_2^{(\beta)}\}$
			 are distinct two-dimensional subspaces from the 
				$2$-spread 
				$\mathcal{P}$.
				Vectors $\boldsymbol{u}_a^{(\alpha)},\boldsymbol{u}_b^{(\beta)}$ are linear independent. By Step 3 of the $\beta$-th round, $\Span\{\boldsymbol{u}_c^{(\gamma)}\}\not\in 
				\left[\Span\{\boldsymbol{u}_a^{(\alpha)},\boldsymbol{u}_b^{(\beta)}\}\atop 1\right]$. 
				Thus,  
				$\boldsymbol{u}_a^{(\alpha)},\boldsymbol{u}_b^{(\beta)},\boldsymbol{u}_c^{(\gamma)}$ are 
				linear 
				independent . The desired result follows.
			\end{proof}
			
			\begin{theorem}\label{thm2}
				The output of Algorithm  \ref{alg1} is a set of vectors:
				$\{\boldsymbol{u}_1^{(1)},\boldsymbol{u}_2^{(1)},\boldsymbol{u}_1^{(2)},\boldsymbol{u}_2^{(2)},\dots,\boldsymbol{u}_1^{(L)},\boldsymbol{u}_2^{(L)}\}\subseteq \mathbb{F}_{q}^4$,
				where $L\geq \max\{\lceil \frac{\sqrt{2}}{3}q\rceil,3\}$ when $q\geq 4$. 
				
				By Proposition \ref{prop1} and Lemma \ref{lem5}, there exists LRCs with parameters 
				$(n=3L\geq \max\{3\lceil\frac{\sqrt{2}}{3}q\rceil,9\},k=2L-4,d=7$ or $8,r=2)_q$ attaining the bound \eqref{eq3} when $q\geq 4$. 
				
				In particular, the minimum distance $d=7$ when $n\textgreater q+4$ by Corollary \ref{cor1}, in which case it is an almost optimal LRC and obtains the best possible minimum distance if its other parameters and the field size $q$ are unchanged. 
			\end{theorem}
			\begin{proof}
			
%
%
%
 
Since $|\mathcal{M}_0|=q^2+1$, we have $2\leq L\leq q^2+1$ when Algorithm \ref{alg1} stops. 
 The set family	$\mathcal{M}_0$ has the following form 
 \begin{equation*}\begin{split}
 \left\{\left[\Span\{\boldsymbol{u}_1^{(1)},\boldsymbol{u}_2^{(1)}\}\atop 1\right],\left[\Span\{\boldsymbol{u}_1^{(2)},\boldsymbol{u}_2^{(2)}\}\atop 1\right],\dots, \right.\\
\left.\left[\Span\{\boldsymbol{u}_1^{(L)},\boldsymbol{u}_2^{(L)}\}\atop 1\right],B_1,B_2,\dots,B_{q^2+1-L} \right\}.
			\end{split}
		\end{equation*}
				Since the sets in $\mathcal{M}_0$ are pairwise disjoint, for any $j\textless i\in [L]$ and $a,b\in \{0,1,2\}$,  $\Span\{\boldsymbol{u}_a^{(i)}\}$ and $\Span\{\boldsymbol{u}_b^{(j)}\}\notin B_1,B_2,\dots,B_{q^2+1-L}$. Thus, we have
				\begin{equation}\label{eq6}
				\bigg|\bigcup\limits_{t=1}^{q^2+1-L} \bigg(B_t\cap \left[\Span\{\boldsymbol{u}_a^{(i)},\boldsymbol{u}_b^{(j)}\}\atop 1\right]\bigg)\bigg|\leq q+1-2.
				\end{equation}
				 Since the set family $\mathcal{M}$ must be empty when Algorithm \ref{alg1} stops, for any $t\in [q^2+1-L]$, we have $$\left|B_t\backslash \bigg(\bigcup\limits_{i=2}^{L}\bigcup\limits_{j=1}^{i-1}\bigcup\limits_{a,b\in \{0,1,2\}} \left[\Span\{\boldsymbol{u}_a^{(i)},\boldsymbol{u}_b^{(j)}\}\atop 1\right]\bigg)\right|\textless 3.$$
				 Since $|B_t|=q+1$, we have
				 $$\left|B_t\cap \bigg( \bigcup\limits_{i=2}^{L}\bigcup\limits_{j=1}^{i-1}\bigcup\limits_{a,b\in \{0,1,2\}} \left[\Span\{\boldsymbol{u}_a^{(i)},\boldsymbol{u}_b^{(j)}\}\atop 1\right]\bigg)\right|\geq q+1-2.$$
				 Since $B_1,B_2,\dots,B_{q^2+1-L}$ are pairwise disjoint sets, we have
				  \begin{scriptsize}
				  	\begin{align*}
				 	\left|\bigcup\limits_{t=1}^{q^2+1-L}\Bigg(B_t\cap \bigg( \bigcup\limits_{i=2}^{L}\bigcup\limits_{j=1}^{i-1}\bigcup\limits_{a,b\in \{0,1,2\}} \left[\Span\{\boldsymbol{u}_a^{(i)},\boldsymbol{u}_b^{(j)}\}\atop 1\right]\bigg)\Bigg)\right|\\
				 	\geq (q^2+1-L)(q-1),
				 	\end{align*}
				 \end{scriptsize}
			 thus,
				 \begin{scriptsize}
		\begin{align}\label{eq7}
				 	\left| \bigcup\limits_{i=2}^{L}\bigcup\limits_{j=1}^{i-1}\bigcup\limits_{a,b\in \{0,1,2\}}
				 	\bigcup\limits_{t=1}^{q^2+1-L} \bigg(B_t\cap\left[\Span\{\boldsymbol{u}_a^{(i)},\boldsymbol{u}_b^{(j)}\}\atop 1\right]\bigg)\right|\notag\\
				 \geq (q^2+1-L)(q-1).			
			 	 \end{align}
				  \end{scriptsize}
				 
		Note that  \begin{small}
			\begin{align*}
			 \sum\limits_{i=2}^{L}\sum\limits_{j=1}^{i-1}\sum\limits_{a,b\in \{0,1,2\}}
			\bigg|\bigcup\limits_{t=1}^{q^2+1-L}\bigg(B_t\cap\left[\Span\{\boldsymbol{u}_a^{(i)},\boldsymbol{u}_b^{(j)}\}\atop 1\right]\bigg)\bigg|\\
			\geq \bigg| \bigcup\limits_{i=2}^{L}\bigcup\limits_{j=1}^{i-1}\bigcup\limits_{a,b\in \{0,1,2\}}
			\bigcup\limits_{t=1}^{q^2+1-L} \bigg(B_t\cap
			\left[\Span\{\boldsymbol{u}_a^{(i)},\boldsymbol{u}_b^{(j)}\}\atop 1\right]\bigg)\bigg|.
			\end{align*}
	\end{small}
		By inequalities \eqref{eq6} and \eqref{eq7}, we have
				$$(q^2+1-L)(q-1)\leq 
				\sum_{i=2}^{L}\big(9(i-1)(q-1)\big),$$
				
				$$(q^2+1)(q-1)\leq (q-1)\frac{9(L-1)L}{2}+(q-1)L,$$
				
				$$\frac{2(q^2+1)}{9}\leq L^2-L+\frac{2}{9}L. 
				$$
				
				When $q=4$, we have $L\geq 3$.
				
				When $q\geq 5$, we have  $\frac{2(q^2+1)}{9}\leq L^2$.
				
				Hence $L\geq \max\{\lceil\frac{\sqrt{2}}{3}q\rceil,3\}$ when $q\geq 4$.

				Let  $\mathcal{C}$ denote the corresponding LRC constructed via Lemma \ref{lem5}. Its length $n=3L\geq \max\{3\lceil\frac{\sqrt{2}}{3}q\rceil,9\}$. Its dimension 
				\begin{footnotesize}
				$k=2L-4\leq \frac{2n}{3}-\lceil\log_q\big(q+q(q-1)(\max\{3\lceil\frac{\sqrt{2}}{3}q\rceil,9\}-2)\big)\rceil$\end{footnotesize} by the bound \eqref{eq3}.
				
				When $q=4,5,7,8$,  it is easy to verify that $$\lceil\log_q\big(q+q(q-1)(\max\{3\lceil\frac{\sqrt{2}}{3}q\rceil,9\}-2)\big)\rceil=4.$$
				
				When $q\geq 9$, we have $3\lceil\frac{\sqrt{2}}{3}q\rceil\geq q+4$, so \begin{align*}
					&\lceil\log_q\big(q+q(q-1)(\max\{3\lceil\frac{\sqrt{2}}{3}q\rceil,9\}-2)\big)\rceil\\
				&\geq \lceil\log_q\big(q+q(q-1)(q+4-2)\big)\rceil=4.
				\end{align*}
				Thus, the dimension of $\mathcal{C}$ attains the bound \eqref{eq3} when $q\geq 4$.
			\end{proof}
				%
				%
			By Algorithm \ref{alg1} and Lemma \ref{lem5}, we give two examples below with the help of the computer algebra system Magma\cite{bosma1997magma}. The minimum distances of these LRCs are also verified by Magma\cite{bosma1997magma}.
			\begin{example}\label{exma1}
				When $q=4$, let $\gamma\in \mathbb{F}_{4}$ be a root of $x^2+x+1\in 
				\mathbb{F}_2[x]$, then $H_1$ is a 
				parity-check matrix of an almost optimal LRC with parameters
				$(9,2,7,2)_4$. By Corollary \ref{cor1}, this LRC obtains the best possible minimum distance if its other parameters and the field size $q$ are unchanged.
				\setlength{\arraycolsep}{4.0pt}
				$$H_1=\left[
				\begin{array}{ccc|ccc|ccc}
					1&1&1&0&0&0&0&0&0\\
					0&0&0&1&1&1&0&0&0\\
					0&0&0&0&0&0&1&1&1\\
					\hline
					1&\gamma^{2}&0  &0&\gamma^{2}&0      &0&\gamma^{2}&0\\
					0&\gamma&0   &1&0&0    &0&0&0\\
					0&0&0    &0&1&0     &\gamma&\gamma^{2}&0\\
					0&\gamma&0    &0&\gamma&0      &0&\gamma^{2}&0\\
				\end{array}\right]$$
			\end{example}
		
			\begin{example}\label{exam2}
				When $q=7$, let $\mathbb{F}_7=\{0,1,2,3,4,5,6\}$. We obtained an
				$(18,8,7,2)_7$ 
				almost optimal LRC with $H_2$ as a parity-check matrix. By Corollary \ref{cor1}, this LRC obtains the best possible minimum distance if its other parameters and the field size $q$ are unchanged.
				\setcounter{MaxMatrixCols}{20}
				\setlength{\arraycolsep}{3.0pt}
				$$H_2=\left[
				\begin{array}{ccc|ccc|ccc|ccc|ccc|ccc}
					1&1&1&0&0&0&0&0&0&0&0&0&0&0&0&0&0&0\\
					0&0&0&1&1&1&0&0&0&0&0&0&0&0&0&0&0&0\\
					0&0&0&0&0&0&1&1&1&0&0&0&0&0&0&0&0&0\\
					0&0&0&0&0&0&0&0&0&1&1&1&0&0&0&0&0&0\\
					0&0&0&0&0&0&0&0&0&0&0&0&1&1&1&0&0&0\\
					0&0&0&0&0&0&0&0&0&0&0&0&0&0&0&1&1&1\\
					\hline
					0&5&0  &6&1&0  &5&3&0  &2&3&0  &5&4&0  &4&1&0\\
					2&3&0  &4&2&0  &5&4&0  &2&5&0  &2&6&0  &0&2&0\\
					0&6&0  &4&3&0  &6&2&0  &2&1&0  &3&0&0  &6&6&0\\
					2&1&0  &0&3&0  &0&4&0  &0&0&0  &3&2&0  &6&5&0\\
				\end{array}\right]$$
			\end{example}
			
			\section{Conclusions}
			In this paper, we focus on LRCs with minimum distance $7$ and locality $2$. 
			We prove an upper bound on the code length of almost optimal LRCs with $d=7$, $r=2$ and disjoint local repair groups, which is better than the bounds proved in \cite{guruswami2019long} and \cite{chen2020improved} in this case. 
			Almost optimal LRCs with $d=7$, $r=2$ and length $n\geq 3\lceil\frac{\sqrt{2}q}{3}\rceil$ are constructed by an algorithm.
			
			However, there is still a gap between the code length of our construction and the upper bound on the code length. More works need to be done to improve the upper bound or the construction.
			\bibliography{Ref}{}
			 \bibliographystyle{IEEEtran} 
		\end{multicols}
	\end{document}